\documentclass{article}

\usepackage{hyperref}

\usepackage{amsmath,amssymb,amsthm}

\usepackage[utf8]{inputenc} 
\usepackage[T1]{fontenc}    
\usepackage{hyperref}       
\usepackage{url}            
\usepackage{booktabs}       
\usepackage{amsfonts}       
\usepackage{nicefrac}       
\usepackage{microtype}      
\usepackage{xcolor}         

\usepackage{graphicx}

\usepackage{authblk}
\newtheorem{theorem}{Theorem}
\newtheorem{definition}[theorem]{Definition}
\newtheorem{lemma}[theorem]{Lemma}
\newtheorem{corollary}[theorem]{Corollary}

\title{Reconstructing Arbitrary Trees from Traces in the Tree Edit Distance Model}

\author{Thomas Maranzatto}
\affil{Department of Mathematics, Statistics, and Computer Science, University of Illinois at Chicago}

\begin{document}

\maketitle

\begin{abstract}
In this paper, we consider the problem
of reconstructing trees from traces
in the tree edit distance model.
Previous work by Davies~et~al.~\cite{DaviesRR19}
analyzed special cases of reconstructing
labeled trees.  In this work, 
we significantly expand our understanding of this problem by giving general results
in the case of arbitrary trees.
Namely, we give:
a reduction from the tree trace reconstruction problem to the more classical 
string reconstruction problem when the tree
topology is known, a lower bound for learning
arbitrary tree topologies, and a general
algorithm for learning the topology of any tree 
using techniques of Nazarov and Peres~\cite{NazarovP17}.
We conclude by discussing why
arbitrary trees require exponentially many
samples under the left propagation model.
\end{abstract}

\section{Introduction}
 In this paper, we study the complexity of
reconstructing trees from traces under the 
tree edit distance model, as defined by
Davies~et~al.~\cite{DaviesRR19}. \emph{Tree trace reconstruction} is a generalization of the traditional \emph{trace reconstruction} problem, where we observe noisy samples of a binary string after passing it through a deletion channel several times~\cite{BatuKKM04}.  
Trace reconstruction and tree trace reconstruction
have a variety of applications including in  archival DNA storage, sensor networks, and linguistic reconstruction \cite{BhardwajDNA, ChurchDNA, DNAStar, bouchard-cote_automated_2013}.

\subsection{String trace reconstruction}
Below we define the string trace reconstruction problem.

\begin{definition}[{string trace reconstruction}]
In the string trace reconstruction problem, a trace of a string $s$ on $n$ binary bits is obtained by running the string
through a deletion channel that removes each bit
of $s$ independently with probability $q$.  Each trace is generated independently of the other traces.
The goal is to reconstruct $s$ with probability at least $1-\delta$ using
as few traces as possible. Denote by $T(n, q, \delta)$ the minimum number of traces needed to reconstruct a string with $n$ bits with deletion rate $q$ and success probability $1 - \delta$.  We sometimes hide the dependency on $q$ by writing $T(n, \delta)$ when $q$ is known.
\end{definition}

This problem has been studied extensively, yet there is still an exponential gap in the upper and lower bounds on trace sample complexity. Chase~\cite{chase2020upper} recently proved that $\exp(\tilde{O}(n^{1/5}))$ traces are sufficient to reconstruct a string, improving on the previous bound of $\exp(O(n^{1/3}))$ traces that was independently discovered by Nazarov and Peres~\cite{NazarovP17} and De~et~al.~\cite{DeOS19}. 

Many variants on trace reconstruction have been considered.  In the average-case setting, the unknown input string is chosen uniformly at random, and improved upper bounds are known in this 
case~\cite{BatuKKM04,HoldenL18,HoldenPP18,HolensteinMPW08}.
In particular, Holden~et~al.~\cite{HoldenPP18} showed that $O(\mathrm{exp}(C \log^{1/3}{n}))$ traces are sufficient to recover the string.  

Other settings which yield improved bounds are in the smoothed-analysis setting where the unknown string is perturbed before generating traces \cite{chen2020perturbed}, in the coded setting where the string is confined to a pre-determined set \cite{brakensiek2019coded,Cheraghchi2020coded,srinivasavaradhan2018coded}, and in the approximate setting where the goal is to output a string that is close to the original in edit-distance \cite{davies2020approximate}.

While theoretically exciting, there is also a direct application of trace reconstruction to DNA data storage where algorithms are deployed to recover the data from noisy samples \cite{BhardwajDNA,ChurchDNA,organickDNA}.  Clearly there is a practical need to reduce the sample complexity as this impacts the time and cost of retrieving archival data stored with DNA. For example, \cite{ChurchDNA} were able to store a 53,426 word book, 11 JPG images, and a JavaScript program (a total of 5.27 megabits) in a DNA medium over 9 years ago.  Finding efficient reconstruction algorithms are crucial to the scalability of this storage medium.

\subsection{Tree trace reconstruction}

Recent work has allowed the creation of robust branching DNA structures for storage \cite{DNAStar, DNABranch}. This naturally leads to the tree trace reconstruction problem posed by Davies~et~al.~\cite{DaviesRR19}.   In their model, a tree of known topology has $n$ nodes, and each node $v$ also carries one hidden bit of information $\ell(v) \in \{0,1\}$.  
This corresponds to each position in a string carrying one bit of information.
They gave reductions from tree trace reconstruction to string trace reconstruction in three specific cases.

\begin{definition}[{tree trace reconstruction}]
In the tree trace reconstruction problem, a trace of a tree $T$ on $n$ nodes is obtained by removing each node 
of $T$ independently with probability $q$.  Each trace is generated independently of the other traces.
The goal is to reconstruct $T$ with probability at least $1-\delta$ using
as few traces as possible.
\end{definition}

The notion of removing a node from $T$ is ambiguous, and there are two main models that have been studied by Davies~et~al.~\cite{DaviesRR19}. In this paper we focus primarily on the \emph{tree edit distance} (TED) deletion model: when a node is deleted, all of its children  become children of its parent.  The deleted node's children take its place as a contiguous subsequence in the left-to-right order.
Another model we briefly consider is called \emph{left-propagation}.  There, when a node is deleted, recursively replace every node (together with
its label) in the left-only path starting at the node with its child in the path. This results in the
deletion of the last node of the left-only path, with the remaining tree structure unchanged.

The tree trace reconstruction problem trivially generalizes the string reconstruction problems; for example, a path graph is a tree, and
under node deletions is equivalent to the string reconstruction problem.  Some tree structures simplify the problem;
for example, 
Davies~et~al.~\cite{DaviesRR19} showed significant improvements over the 
$O(\mathrm{exp}({C n^{1/3}}))$ bound for the cases of $k$-ary trees and spider graphs.
However, it was \emph{a priori} possible that some tree structures may present harder problems and require even
more samples than their string counterparts.

In addition to the interesting theoretical
questions this model raises, it also has applications
to real-world problems in constructing branched DNA structures.  This topic is relevant for long-term data structure storage and progress on it could help
to improve data density in the storage medium. 

For example, He et al. \cite{DNAStar} were able to synthesize short-armed star DNA structures with 3 to 12 arms in total.  Their process allowed for control over how many arms are produced, and they were able to detect the number of arms in a given sample. Furthermore, Karau and Tabard-Cossa \cite{DNABranch} investigated T and Y shaped DNA strands, where the backbone of the molecule is twice as long as the branches. They provide biochemical methods for robustly detecting these structures.  Creating tree trace reconstruction algorithms is then crucial for the viability of branched DNA as an archival storage option.






\subsection{Our results}

In this paper we give the following results.
\begin{itemize}
\item In section~\ref{sec:treestotraces},
we give a reduction from the
problem of reconstructing trees from traces
to the problem or reconstructing strings from 
traces, in the case where the tree structure is known.
Our reduction only increases the running time by a multiplicative 
factor of $n$.\footnote{This answers Question 2 in Section 6.1 of
Davies~et~al.~\cite{DaviesRR19}.}
\item In Section~\ref{sec:badtopology}, we give a 
lower bound for learning the topology of a tree from tree traces.  In particular, we construct a ``linked list'' tree with $n$ nodes and simulate
a string reconstruction problem by adding 
leaves below each of these nodes, either to the left or to the right of the main chain.  Learning traces from this tree is very similar to learning string traces.
\item In Section~\ref{secfuzzy}, We give an 
algorithm for learning traces using $\exp(\tilde{O}(n^{1/5}))$ samples.
There, we use the argument from Chase~\cite{Chase19lower}, but applied to our flattened tree traces. This bound is independent of learning the tree labels.
\item In Section~\ref{sec:Left-Propagation}, we show that $O((1-q)^{-n})$ samples are needed to reconstruct a tree whose traces are generated from the left-propagation deletion channel.
\end{itemize}

\section{Related Work}
The problem of string-trace reconstruction has received considerable attention since it was first introduced by \cite{BatuKKM04}.  Despite this there is still an exponential gap in the necessary number of traces and the sufficient number of traces needed to reconstruct an arbitrary input string \cite{Chase19lower,chase2020upper,DeOS19,HoldenL18,NazarovP17}.  The best lower bound thus far is due to Chase~\cite{Chase19lower}, where
he showed that $\tilde{\Omega}(\log^{5/2}n)$ traces are necessary, and
the best upper bound is due to Chase~\cite{chase2020upper},
where the same author showed that $\exp (\tilde{O}(n^{1/5}))$ traces are sufficient.  Previous to this, Nazarov and Peres~\cite{NazarovP17} and De~et~al.~\cite{DeOS19} independently achieved the upper bound of $\exp (O(n^{1/3}))$ traces.

Variants of trace reconstruction have been proposed in order to develop new techniques to close the exponential gap in the original problem formulation.  Davies~et~al.~\cite{DaviesRR19} investigated a new problem of reconstructing trees from their traces.  The authors investigated two classes of trees;  complete $k$-ary trees and spider graphs.  They also introduced the ``tree edit distance'' (TED) and ``left-propagation'' deletion channels.

For a complete $k$-ary tree $T$ the authors considered the general case for arbitrary $k$, as well as when $k$ is `large' compared to the total number of nodes.  Under TED, if $k \geq O(\log^2(n))$, then a complete $k$-ary tree can be reconstructed using $\exp(O(\log_k n)) \cdot T(k, 1/n^2)$ traces generated from $T$.  With  no restrictions on $k$, $\exp(O(k \log_k n)) $ traces suffice to reconstruct $T$.  Under left-propagation, if $k \geq O(\log \ n)$ then $T(O(\log_k n + k), 1/n^2))$ traces suffice to reconstruct T.  For arbitrary $k$  then $O(n^\gamma \log n)$ traces suffice, where $ \gamma=\ln \left(\frac{1}{1-q}\right)\left(\frac{c^{\prime} k}{\ln n}+\frac{1}{\ln k}\right)$. The authors ues primarily combinatorial arguments to achieve these bounds, which is in contrast to the current methods used in string trace reconstruction.  It was given as an open problem to extend these combinatorial arguments to general trees.  

The case of spider graphs was broken into small leg-depth and large leg-depth.  Note that for spider graphs, TED and left-propagation  are identical deletion channels. For depth $d \leq \log_{1/q} n $ and $q > 0.7$, then $\exp (O(d(nq^d)^{1/3})) $ traces suffice to reconstruct the spider.  For depth $d \geq \log_{1/q} n$ and all $q > 1$, then $2 T(d, 1/2n^2)$ traces suffice to reconstruct the spider. The algorithms proposed here generalize the mean-based methods of De~et~al.~\cite{DeOS19} and Nazarov and Peres~\cite{NazarovP17} to multiple independent strings, where some strings have a chance of being completely removed.

\section{Learning trees of known topology}\label{sec:treestotraces}

In this section, we show that it is indeed not the case that tree reconstruction can require more traces than the corresponding string reconstruction.  
In particular, in the case that the tree structure is known,
we reduce the problem of learning trees from traces with no overhead in the number of traces required
and only a multiplicative linear overhead in the time complexity.  We show this by giving a traversal of the tree that converts the tree to a string
in linear time such that node deletions  occur in their corresponding positions in the string.  
This allows string trace reconstruction algorithms to be used for tree trace reconstruction.


We start by giving the main lemma of this section, which (informally) shows that the order of nodes visited by an pre-order traversal
of a tree is preserved under deletions.  We note that a pre-order traversal is usually defined for binary trees; we define a pre-order traversal of general trees as recursively visiting the root first and then
all children left to right.

\begin{lemma}\label{lem:traversal}
Let $T$ be a tree and $T'$ be a tree obtained from $T$ via an arbitrary set of deletions under either the trace edit 
distance or the left propagation model.  
If $u, v$ are any two nodes in $T'$ such that an pre-order
traversal of $T'$ visits $u$ before $v$ then a pre-order traversal of $T$ will also visit $u$ before visiting $v$.
\end{lemma}
\begin{proof}
First, we consider the tree edit distance model.
Given the deletion of an arbitrary node $w$, we consider how it effects the traversal order of the remaining nodes.
Figure~\ref{fig:trees} shows the general case when a node $w$ is removed, with the subtrees numbered in the order they are visited according to a pre-order traversal.

\begin{figure}[h]
\begin{center}
\includegraphics[width=3.2in]{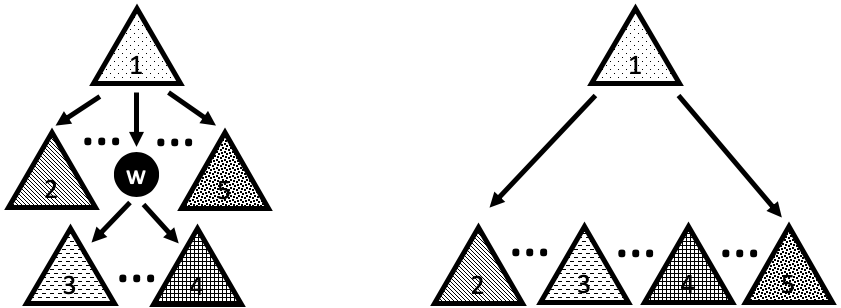}
\end{center}

\caption{The generic picture before (left) and after (right) node $w$ is removed.  
The subtrees trees are labeled in the order they are visited in a pre-order traversal. Note that on the left, $w$ is visited immediately 
before tree labeled with the number $3$.}\label{fig:trees}
\end{figure}

If $u$ and $v$ fall in the same subtree, their visitation order is clearly not affected.  Otherwise, if $u$ is in a tree that's visited
earlier than when the tree of $v$ is visited on the left of Figure~\ref{fig:trees}, then it is also visited earlier on the right.
This simple illustration captures all of the cases (i.e.\ pairs of trees where $u$ and $v$ can occur), which finishes the first part of the proof.
\end{proof}

Let $T'$ be a trace of $T$ and let $v_1 \ldots v_{n'}$ be the nodes of $T'$ indexed by the order they are visited in a pre-order traversal.
Let 
$$s(T') = \ell(v_1)\ell(v_2) \ldots \ell(v_{n'}).$$
Let $T'_1, T'_2, \ldots T'_m$ be traces of $T$.  This procedure allows us to obtain $m$ strings $$s(T'_1), s(T'_2), \ldots s(T'_m).$$ 
Using Lemma~\ref{lem:traversal}, if $m$ is sufficiently large, we can run a string trace reconstruction algorithm on this sequence of strings $m$ to obtain
the string $s(T)$, which allows us to label $T$'s nodes in the order given by a pre-order traversal.

Hence, we can conclude the following.

\begin{theorem}
Let $T$ be any tree whose topology is known, its labels can be reconstructed using $T(n, q, \delta)$ traces in the worst case.
\end{theorem}

Using the reduction above and applying the results of Chase~\cite{chase2020upper} for recovering any string from  $\tilde{O}(\mathrm{exp}(n^{1/5})$ traces, we get the following.

\begin{corollary}\label{cor:knowntopo}
Let $T$ be any tree whose topology is known, its labels can be reconstructed using $\tilde{O}(\mathrm{exp}(n^{1/5})$ traces in the worst case.
\end{corollary}

Similarly, applying the known results of 
Hartung~et~al.~\cite{HartungHP18}
 for recovering a random string from $O(\mathrm{exp}({C \log^{1/3}n}))$ traces, we obtain the following
as a corollary.

\begin{corollary}
Let $T$ be any tree whose topology is known, if $T$'s nodes are labeled uniformly at random, its labels can be reconstructed using  
$O(\mathrm{exp}(C \log^{1/3}n)$ traces.
\end{corollary}

A natural open question is what happens in trees whose topology is not known.  In that case, it makes sense to attempt to 
recover the topology.

\section{A lower bound on learning the topology}\label{sec:badtopology}

One may then ask about the problem of reconstructing the topology of an unknown tree. For this problem, we can assume all the labels are the same, as to give the least possible information to the learner. We call such a tree an unlabeled tree.

For the lower bound we show there exists an unlabelled tree that naturally encodes the string trace reconstruction problem. Given a string $S$ we construct the tree $T_S$ in two stages. First, create a path tree with $|S| + 2\ell$ vertices, where $\ell = \Theta\left( \frac{\ln (1/\delta) +  \ln(T(n, \delta))}{\ln(1 / q)}\right)$. Next, for each bit $S_i \in S$, if $S_i = 0$ add a left child to the node in position $(\ell+i)$ in the path. Otherwise $S_i = 1$ so add a right child to node $(\ell+i)$ in the path. Call the subtree which includes left and right leaves the \textbf{encoding} of $S$. We now show that this family of trees is sufficiently hard to learn under the TED model.

\begin{lemma}
Let $T_S$ be a tree constructed from a string $S$ using the process defined above. Under the TED deletion channel $T_S$ requires $\Omega(T(n, q^2, \delta))$ traces for reconstruction.
\end{lemma}
\begin{proof}
By the Hoeffding bound, we can recover the longest path in $T_S$ with probability $1-\delta$ using $O(n \log(1/ \delta))$ samples. This gives the `backbone' of the string encoding. Choosing $$\ell = \Theta\left( \frac{\ln (1/\delta) +  \ln(T(n, \delta))}{\ln(1 / q)}\right)$$ ensures that with probability $1 - \delta$ in all $T(n, \delta)$ independent trials, at least one of the vertices survive in both the upper and lower `buffer sections' of the tree.
We know which positions in the backbone the leaves will occur in $T_S$, thus it suffices to learn the orientation of the leaves with high probability. 
\begin{figure}[h!]
\begin{center}
\includegraphics[width=.8in]{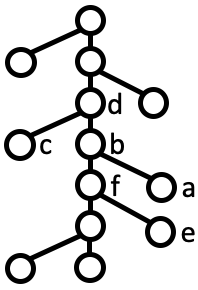}
\end{center}
\caption{Example configuration of a leaf node $a$ and the leaves nearest it. The survival of $a$ is dependent on the nodes in its 3-neighborhood.}\label{lowerb}
\end{figure}

For this proof, we refer to Figure~\ref{lowerb} which shows a possible configuration of the main path and leaves oriented on either side. The exact orientations of the leaves do not matter in the following analysis. Say node $a$ is completely removed in a trace if both it and it's parent are deleted from the $T_S$.  The probability that $a$ is completely removed in a trace is dependent on its parent being removed as well. In general there may be many ways leaves could be arranged in a trace if they are not completely deleted. Suppose there existed some oracle that given a trace generated from $T_S$ could replace `backbone' nodes and leaf nodes that were only partially deleted. Could we hope to reconstruct the original tree with fewer traces using this oracle? Here, nodes are removed with probability $q^2$. Thus even an algorithm endowed with this oracle cannot place these leaves in the tree, so leaves survive with probability $1 - q^2$. 

There is a 1-to-1 correspondence between leaf position in the tree and bit value in a string. Further, the survival of a bit with the oracle learner is $O(p^2)$, so the final algorithm would require $\Omega(T(n, q^2, \delta))$ traces to reconstruct $T_S$.  This is the best any algorithm could hope to do, as nodes that are completely removed by definition do not leave any information in the trace.  

\end{proof}

This gives a natural lower bound to the tree trace reconstruction problem, as we have no guarantee on the input tree to the problem.
\begin{theorem}
For any two trees $T_R$ and $T_S$ encoding binary strings $R$ and $S$ respectively, then $\Omega(T(n, q^2, \delta))$ traces are required to distinguish between $T_R$ and $T_S$.
\end{theorem}

Using the above argument along with the results of Chase~\cite{chase2020upper}, it follows $T_S$ can be reconstructed using $\tilde{O}(\mathrm{exp}(n^{1/5})$ traces.

\begin{corollary}
\label{uppern13}
Let $T_S$ be a tree encoding of binary string $S$. Then for any fixed deletion probability $q$,  $\tilde{O}(\mathrm{exp}(n^{1/5})$ traces suffice to reconstruct $T_S$ with probability $1- \delta$.
\end{corollary}

\section{Recovering fuzzy trees of degree $\tilde{O}(n^{1/5})$}
\label{secfuzzy}
For any tree, call a leaf terminal if all of its siblings are leaves.  We define a \emph{fuzzy tree of degree m} as a tree where all terminal leaves have $m-1$ siblings.

Recall the result from \cite{Chase19lower} that gives $T(n, \delta) \leq \exp(\tilde{O}(n^{1/5}))$.

\begin{theorem}
\label{fuz}
Given a fuzzy tree $A$ of degree $\tilde{O}(n^{1/5})$ and probability $\delta$, we can reconstruct $A$ using $\tilde{O}(n^{1/5})$ traces with high probability.
\end{theorem}
\begin{proof}
Sample $T(n, \delta)$ traces from $A$.  For trace $m$, construct two binary strings $S^m_0, S^m_1$ from a preorder traversal of the trace. $S^m_0$ uses the alphabet $\{0, 2\}$ and $S^m_1$ uses the alphabet $\{2, 1 \}$. $S^m_1$ is constructed by writing a 1 when the traversal moves down an edge, and a 2 when the traversal encounters a leaf.  $S^m_2$ is constructed by writing a 0 when the traversal moves up an edge, and a 2 when the traversal encounters a leaf.

Because of the assumption that $A$ is fuzzy, for every trace the probability that there exists a node with all children deleted is less than $1/(n \cdot exp(\tilde{O}(n^{1/5}))$.  By the union bound the probability that any trace contains a node where all children are deleted is less than $1/n$.

$A$ has two unknown strings $S_0$ and $S_1$ constructed analogously to the trace strings above.  It is easy to see that a node deletion in $A$ corresponds to a single bit deletion in $S_0$ and $S_1$.  This property holds for multiple deletions as well. Further with high probability no node will have all its children removed.  Because only one edge is represented by each bit in the two strings, only one edge is removed in a TED deletion, and all children of a leaf are not deleted with high probability, we conclude that with high probability deletions in $S_0$ and $S_1$ behave as $i.i.d$ deletions in the string deletion channel model.  Therefore, by applying the results from \cite{Chase19lower} we obtain the desired result.
\end{proof}

\section{Recovery under Left-Propagation}\label{sec:Left-Propagation}

Here we show that learning under left propagation requires exponentially many samples.  Namely we show the existence of two simple trees $A$ and $B$ that can only be distinguished if the left propagation deletion model removes no nodes.  This happens with probability $(1-q)^n$, so the worst-case sample bound would then be $\Omega((1-q)^{-n})$.  This bound is also what any nontrivial algorithm should try to overcome.

We call $LP_k(A)$ the $k$'th order left propagation trace set of tree $A$, which is the set of all traces obtainable by removing $k$ nodes from $A$. The trace set contains no probabilistic information, it simply categorizes what traces could be generated regardless of how unlikely they are to appear.
We approach this by showing the trace sets of $A$ and $B$ under left-prop are identical for any non-zero number of deletions.  We use the shorthand $A_n$ and $B_n$ to show that $A$ and $B$ both have $n$ non-root nodes.

\begin{figure}[h!]
\begin{center}
\includegraphics[width=1.4in]{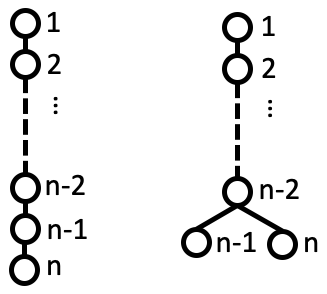}
\end{center}
 \caption{The two trees we are investigating.  Notice that the trees' structures are identical until node $n-2$. Removing any single node in either tree results in the tree $A_{n-1}$.}
\end{figure}

\begin{theorem}
For all $n > 3$ and all $m \leq n$ there exist unlabelled trees $A$ and $B$ containing $n$ non-root nodes each, such that $LP_m(A) = LP_m(B)$. 

\end{theorem}
\begin{proof}
For $n > 3$, consider the unlabelled tree $A$ with $n$ nodes, where every internal node has exactly one child, including the root.  Also consider the tree $B$ with $n$ nodes which is constructed by taking $A_{n-1}$ and adding a sibling to the single leaf.

The set $LP_1(A_n)$ only contains one item, namely $A_{(n-1)}$.  There is no dependence on a deletion probability here, we are requiring that one and only one node is removed from the tree.  Set $m < n$. Using the above as a base case, by recursion we see that $LP_m(A_n) = A_{(n-m)}$.

An analogous argument holds for $LP_1(B_n)$. If any node in the chain is deleted, then the left leaf is promoted to be the parent of the right leaf.  This structure is a long chain with $n-1$ non-root nodes, the definition of $A_{(n-1)}$.  If either leaf is removed, then we obviously get a chain of $n-1$ non-root nodes. Thus $LP_1(B_n) = A_{(n-1)}$  Because of this the same recursive argument for $LP_m(A_n)$, implying $LP_m(B_n) = A_{(n-m)}$.
So we have 
\begin{align*}
LP_m(A_n) & = LP_m(B_n) \\
& = A_{(n-m)}
\end{align*}
finishing the proof.
\end{proof}
An easy corollary follows from above.
\begin{corollary}
In the worst case, $\Omega((1-q)^{-n})$ traces are required to  learn a tree topology in the Left Propagation model.
\end{corollary}

\section{Discussion}
In this paper we gave many connections between
the string trace reconstruction and tree trace
reconstruction problems.  Our results
exploited these connections and string
trace reconstruction techniques
to give state-of-the-art bounds for reconstructing
arbitrary trees from traces.

\bibliography{paper}

\end{document}